\documentclass{article}

\usepackage{spconf,amsmath,graphicx,hyperref}
\usepackage[active]{srcltx} 
\usepackage{amsfonts}
\usepackage{amsbsy}
\usepackage{amssymb}
\usepackage{url}
\usepackage{enumerate}
\usepackage{cite}
\usepackage{psfrag}
\usepackage{color}
\usepackage{euscript}
\usepackage[utf8]{inputenc}
\usepackage{algorithmic}
\usepackage[ruled,vlined]{algorithm2e}
\usepackage{xfrac,bigints}
\usepackage{float}
\usepackage{pgf}
\usepackage{amsthm}
\usepackage{subcaption}
\usepackage{mwe}
\usepackage{footmisc} 

\usepackage{tikz}
\usepackage{textcomp}
\usepackage{lipsum}
\usepackage[T1]{fontenc}
\usepackage{textcomp}
\usepackage{multirow}
\usepackage{anyfontsize}

\newcommand{\E}{{E}}

\newcommand{\eqdef}{\triangleq}

\newcommand{\intens}{\EuScript{I}}
\newcommand{\pot}{\EuScript{P}}

\DeclareMathAlphabet{\pazocal}{OMS}{zplm}{m}{n}

\def\bdm#1\edm{\begin{displaymath}#1\end{displaymath}}
\def\be#1\ee{\begin{equation}#1\end{equation}}
\def\barr#1\earr{\begin{align}#1\end{align}}
\newtheorem{proposition}{Proposition}

\newcommand{\IeeeTAP}{{\em IEEE Trans.\ Antennas Propag.\/}}

\title{Optimal transmit field distribution for partially obstructed \\ continuous 
radiating surfaces in near-field
communication systems
}
\name{Francesco Verde$^{1}$, 
Donatella Darsena$^{2}$, 
Marco Di Renzo$^{3,4}$, and
Vincenzo Galdi$^{5}$
\thanks{The work of D.~Darsena and V.~Galdi was partially supported by 
the European Union-Next Generation EU under the Italian
National Recovery and Resilience Plan (NRRP), Mission 4,
Component 2, Investment 1.3, CUP E63C22002040007, partnership
on ``Telecommunications of the Future" (PE00000001
- program ``RESTART"). 
The work of M. Di Renzo was supported in part by the European Union under the Horizon Europe projects COVER (101086228), UNITE (101129618), INSTINCT (101139161), and TWIN6G (101182794); by the Agence Nationale de la Recherche (ANR) through the France 2030 project Networks of the Future (ANR-PEPR NF-SYSTERA 22-PEFT-0006); by the CHIST-ERA project PASSIONATE (CHIST-ERA-22-WAI-04/ANR-23-CHR4-0003-01); and by the UK Engineering and Physical Sciences Research Council (EPSRC) and the Department for Science, Innovation and Technology (EP/X040569/1).
}
}
\address{$^{1}$University of Campania ``Luigi Vanvitelli'', I-81031 Aversa, Italy \\
 $^{2}$University of Naples Federico II, I-80125 Naples, Italy\\
$^{3}$Universit\'e Paris-Saclay, CNRS, CentraleSup\'elec,  91192 Gif-sur-Yvette, France, \\
$^{4}$King's College London, WC2R 2LS London, United Kingdom \\
$^{5}$University of Sannio, I-82100 Benevento, Italy
}
\begin{document}
\maketitle
\vspace{-2mm}
\begin{abstract}
\vspace{-2mm}
This paper deals with the optimal synthesis of aperture fields for (radiating) near-field communications in obstructed environments. A physically consistent model based on knife-edge diffraction is used to formulate the problem as a maximization in Hilbert space. The optimal solution is obtained as a matched filter that ``matches" the shape of a diffraction-induced kernel, thus linking wave propagation with signal processing methods. 
The framework supports hardware implementation using continuous apertures such as metasurfaces or lens antennas. 
This approach bridges physically grounded modeling, signal processing, and hardware design for efficient energy focusing in near-field obstructed channels.
\end{abstract}

\vspace{-2mm}
\section{Introduction}
\label{sec:intro}
\vspace{-2mm}

Continuous aperture arrays have recently emerged as a promising paradigm for next-generation wireless communications \cite{Liu-2025}, enabling the exploitation of nearly continuous electromagnetic (EM) apertures to approach the fundamental physical limits of spatial degrees 
of freedom \cite{Mikki-2023}. Unlike conventional spatially discrete antenna arrays, 
their continuous counterparts lead to an integral representation of the propagation channel, allowing for fine-grained control of amplitude and phase distributions over the entire aperture surface. 
Recently, hardware implementations of continuous aperture arrays
based on lens antennas \cite{Lee2025} or metasurfaces composed of densely spaced subwavelength elements  \cite{Malevich.2025} have been developed, thus making continuous apertures practically feasible.
Such a shift has motivated the development of new signal processing formulations that depart from classical matrix-based 
multiple-input multiple-output (MIMO) models and rely instead on operator-theoretic approaches.

In near-field communication scenarios \cite{Petrov-2024, Cui-2023,Singh-2025,Dar_ArXiV2025}, sharp structural discontinuities such as building edges or panels can partially obstruct the line-of-sight path (LoS). To model the resulting diffraction effects in a physically consistent yet tractable way, this work adopts the  knife-edge model \cite{Goodman,Orfanidis-2002}, which captures wavefront distortion and energy redistribution through a closed-form diffraction kernel. This allows accurate representation of obstruction effects while preserving analytical simplicity for operator-based beamforming design.

In such a non-LoS (NLoS) context, optimal beamforming and energy focusing  in the (radiating) near-field region can be formulated as an optimization problem in Hilbert space, where the optimal aperture field is the filter matched to the diffraction-induced kernel. 
The resulting framework provides a unified perspective that connects diffraction physics,  optimization, and hardware-oriented beamforming design
in the near-field region, paving the way for robust near-field communication schemes in NLoS propagation environments.

\vspace{-2mm}
\section{Knife-Edge Diffraction}
\label{sec:knife-diffraction}
\vspace{-2mm}

With reference to  Fig.~\ref{fig:fig_1}, in order to evaluate the effect of an obstacle on the evolution of the wavefront, we adopt the {\em knife-edge diffraction} model. In this formulation, the obstacle is modeled as a perfectly conducting edge, infinitesimally thin and sharply defined, positioned at $z_\text{b}>0$ from the transmit aperture plane. 
The obstructing surface is oriented perpendicular to the direct propagation path between transmitter and receiver, extending infinitely in the $y$-direction, and covering the transverse range $x\in[x_\text{b}^{(1)},x_\text{b}^{(2)}]$. This simplified geometry ensures that diffraction arises solely from the sharp edge, which acts as a source of secondary wavelets in accordance with the Huygens–Fresnel principle \cite{Goodman}.
For clarity, we represent the transmitter as a one-dimensional flat aperture extending from $x_\text{a}^{(1)}$ to $x_\text{a}^{(2)}$ along the transverse  $x$-axis, located in the plane $z=0$, and assumed infinite in the  $y$-direction (see Fig.~\ref{fig:fig_1}). This idealization represents a rectangular strip antenna with length along the $y$-axis
much larger than its width $\Delta x_\text{a} \eqdef x_\text{a}^{(2)} 
- x_\text{a}^{(1)}$. The transmit aperture is modeled as a spatially continuous radiating surface, which can be practically realized through lens antennas \cite{Lee2025} or metasurfaces formed by densely packed subwavelength elements \cite{Malevich.2025}.

\begin{figure}[t]
\centering
\includegraphics[width=0.8\linewidth]{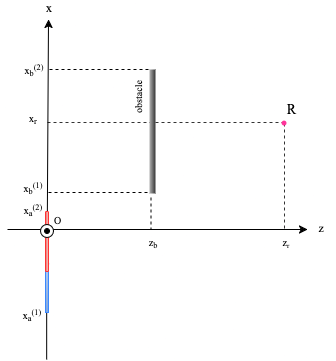} 
\caption{Knife-edge diffraction setup and geometry.}
\label{fig:fig_1}
\end{figure}

As a first step, we describe the propagation process from the aperture plane at  
$z=0$ up to the obstructing plane located at $z=z_\text{b}$. 
In a linear, isotropic, homogeneous, and nondispersive dielectric medium, both 
the electric and magnetic field vectors satisfy the same vector wave equation \cite{Goodman}. Consequently, each Cartesian component obeys an identical scalar wave equation and 
EM propagation can be fully described by a single scalar wave
function.
We specifically examine the case of a $y$-polarized purely monochromatic electric field
\be
\widetilde{\E}(z,x;t) = \Re 
\left\{\E(z,x) \, e^{ j 2 \pi f_0 t} \right\} \:,
\label{eq:rf}
\ee 
propagating along the positive longitudinal  $z$-axis from $z=0$ to $z=z_\text{b}$
and orthogonal to the transverse  $x$-direction, with no variation along $y$, 
where $\Re\{\cdot\}$ denotes the real part, $f_0 >0$ is the carrier frequency,  
$\lambda_0=\tfrac{c}{f_0}$ is the corresponding wavelength in vacuum, 
with $c\approx 3 \cdot 10^8$ m/s being the speed of light in vacuum. 
Let $k_0 = \tfrac{2 \pi}{\lambda_0}$ denote the wave number, 
\be
\E(z,x)=u(z,x) \, e^{-j k_0 z} \:,
\ee
represents the complex envelope of the wave for $0 < z < z_\text{b}$.
Under the {\em paraxial approximation}, i.e., when the observation point lies  close to the propagation axis $z$ \cite{Goodman,Orfanidis-2002}, the field $u(z,x)$ 
in the radiative near field region of the transmitting aperture 
is  given by the {\em Huygens-Fresnel diffraction formula}
\be
u(z,x) = \sqrt{\frac{j}{\lambda_0 z}}  \int_{x_\text{a}^{(1)}}^{x_\text{a}^{(2)}} \E_\text a(\nu) \, e^{-j \frac{k_0}{2 z} (x-\nu)^2} \, \mathrm{d}\nu  \:,
\label{eq:int}
\ee
where $\E_\text a(\nu)$ denotes the aperture field distribution.

The second step pertains the calculation of the complex envelope 
of the diffracted field in the region $z > z_b$.
The baseband field transmitted beyond the obstacle is governed by the Rayleigh–Sommerfeld integral
\be
\E_\text{d}(z,x) = 
\frac{1}{2 j} \int_{\mathcal{O}_\text{d}} 
\frac{k_0 \, (z-z_\text{b})}{\rho_\text{d}(\xi)} \, \E(z_\text{b},\xi)\, 
\text{H}^{(2)}_1(k_0 \, \rho_\text{d}(\xi)) \, \mathrm{d}\xi \:, 
\label{eq:diff-field}
\ee
for $z>z_\text{b}$, where 
$\text{H}^{(2)}_1(\cdot)$ is the first-order second-kind Hankel function
and the integration domain is defined as
\be
\mathcal{O}_\text{d} \eqdef (-\infty,x_\text{b}^{(1)}) \cup (x_\text{b}^{(2)},+\infty)
\ee
and the distance function is
\be
\rho_\text{d}(\xi) \eqdef \sqrt{(x-\xi)^2+(z-z_\text{b})^2} \:.
\ee
In the paraxial regime, i.e., for $|x-\xi | \ll z-z_\text{b}$, and using the large-argument asymptotic expansion of the Hankel function given by (see, e.g., \cite{Orfanidis-2002})
\be
H^{(2)}_1(u) \approx -\sqrt{\frac{2}{j \pi u}} \, e^{-j u} \:,
\label{eq:approx-hank}
\ee
the diffracted field in \eqref{eq:diff-field} reduces, for $z>z_\text{b}$, 
to the form
\be
\E_\text{d}(z,x)=u_\text{d}(z,x) \, e^{-j k_0 z} \:,
\label{eq:field-decomp-diff}
\ee
with the envelope expressed as
\be
u_\text{d}(z,x) = \sqrt{\frac{j}{\lambda_0 (z-z_\text{b})}}
\int_{\mathcal{O}_\text{d}} u(z_\text{b},\xi) \,  e^{-j \frac{k_0}{2 (z-z_\text{b})} (x-\xi)^2} \, \mathrm{d}\xi  \: .
\label{eq:pert-ap}
\ee

By substituting \eqref{eq:int} in \eqref{eq:pert-ap}, after exchanging
the order of integration, one has the compact form
\be
u_\text{d}(z,x)  =
\int_{x_\text{a}^{(1)}}^{x_\text{a}^{(2)}} K(z,x,\nu) \, \E_\text a(\nu)  \, \mathrm{d}\nu
\ee
for $z>z_\text{b}$, where the {\em kernel} $K(z,x,\nu)$ encompasses all propagation and diffraction effects between the aperture coordinate $\nu$ and the observation point $(z,x)$. In the case of a knife-edge obstacle at $z=z_b$, the kernel takes the form
\begin{multline}
K(z,x,\nu) = \frac{j}{\lambda_0 \sqrt{z_\text{b}\,  (z-z_\text{b})}}
\int_{\mathcal{O}_\text{d}} 
e^{-j \frac{k_0}{2 z_\text{b}} (\xi-\nu)^2} 
\\ \cdot 
 e^{-j \frac{k_0}{2 (z-z_\text{b})} (x-\xi)^2} \, \mathrm{d}\xi \:.
\label{eq:K-def}
\end{multline}
The following proposition provides a closed-form expression of the kernel
$K(z,x,\nu) $, which is useful for subsequent optimization purposes. 

\vspace{-1mm}
\begin{proposition}
\label{prop:1}
It results that
\be
K(z,x,\nu) = \frac{1}{2} \sqrt{\frac{j}{\lambda_0 \, z}} \, 
e^{-j \, \frac{k_0}{2 \, z} \, (\nu-x)^2} \, F(z,x,\nu)
\label{eq:prop-K}
\ee
where the knife-edge factor $F(z,x,\nu)$ is defined in \eqref{eq:factor-F}.
\begin{figure*}[!t]
\normalsize
\begin{multline}
F(z,x,\nu) \eqdef
1 + \text{Erf}\left(\sqrt{\frac{j \, k_0 \, z}{2 \, z_\text{b} \, (z-z_\text{b})}} 
\left[ x_\text{b}^{(1)} - \frac{z_\text{b} \, x}{z} - \frac{(z-z_\text{b}) \, \nu}{z}\right]\right)
\\ 
+ \text{Erfc}\left(\sqrt{\frac{j \, k_0 \, z}{2 \, z_\text{b} \, (z-z_\text{b})}} \left[ x_\text{b}^{(2)} 
- \frac{z_\text{b} \, x}{z} - \frac{(z-z_\text{b}) \, \nu}{z} \right]\right) \:.
\label{eq:factor-F}
\end{multline}
\hrulefill
\end{figure*}
\end{proposition}

\begin{proof}
The integral in \eqref{eq:K-def} is Gaussian and can be 
evaluated in closed form using the error function with complex argument \cite{Irmak-2019}.
Details are omitted due to the lack of space.
\end{proof}
It should be noted that, when the upper edge of the obstacle tends to infinity, i.e., 
$x_\text{b}^{(2)}  \to + \infty$, the Erfc in \eqref{eq:factor-F} disappears
and, thus, the expression of $K(z,x,\nu)$ simplifies. 

\vspace{-2mm}
\section{Optimal field on the aperture}
\label{sec:optimization}
\vspace{-2mm}

The determination of the aperture field that maximizes the diffracted intensity at 
the receiver  can be rigorously cast as an optimization problem in a Hilbert space 
framework.  We consider a {\em pointwise receiver}, i.e., the receiver 
is idealized as infinitesimally localized; in this case, the aim is to maximize the intensity 
of the diffracted field  at a single spatial location, 
subject to a power constraint on the aperture.

For a pointwise receiver located at point $R \equiv (z_\text{r},x_\text{r})$, 
the cost function is the field intensity
\be
\intens_\text{r} \eqdef  |u_\text{d}(z_\text{r},x_\text{r}) |^2 =
\left| \int_{x_\text{a}^{(1)}}^{x_\text{a}^{(2)}} K(z_\text{r},x_\text{r},\nu) \, \E_\text a(\nu)  \, \mathrm{d}\nu
 \right |^2  \:.
\label{eq:intensity}
\ee
Physically, in scalar diffraction theory, $\intens_\text{r}$ is related, up to a scalar constant,  
to the time-averaged power flux per unit area carried by the wave. Maximization of 
$\intens_\text{r}$ with respect to
the aperture field $E_\text a(\nu)$ is carried out under the power constraint
\be
\int_{x_\text{a}^{(1)}}^{x_\text{a}^{(2)}} \left| \E_\text a(\nu) \right|^2  \, \mathrm{d}\nu= 
\pot_\text{a} < + \infty \:.
\label{eq:power-constr}
\ee

Maximization of  $\intens_\text{r}$ with respect to $E_\text a(\nu)$ subject to constraint 
\eqref{eq:power-constr} can be readily carried out in the Hilbert space $L^2(\mathcal A)$
of square-integrable complex aperture field distributions over the aperture domain 
$\mathcal{A} \eqdef [x_\text{a}^{(1)},x_\text{a}^{(2)}]$. In such a space,  the inner product is defined as 
\be
\langle f, g \rangle \eqdef \int_{x_\text{a}^{(1)}}^{x_\text{a}^{(2)}} f^*(\nu) \, g(\nu) \, 
\mathrm{d}\nu
\ee
whose associated norm is $\|f\| \eqdef \sqrt{\langle f, f \rangle}$.
The propagation kernel $K_\text{r}(\nu) \eqdef K(z_\text{r},x_\text{r},\nu)$ 
associated with the receiver point $R$ defines a bounded linear 
functional on $L^2(\mathcal A)$. Therefore,  by virtue of
\eqref{eq:int},  we may write  
$\intens_\text{r} = \left |\langle K_\text{r}^*,  \E_\text a \rangle \right|^2$
and invoke the Cauchy-Schwarz inequality 
$\intens_\text{r} \le \pot_\text{a} \, \| K_\text{r}\|^2$, 
where equality holds if and only if $\E_\text a(\nu)$ is proportional to the complex conjugate of
$K_\text{r}(\nu)$. Therefore, accounting for the power constraint \eqref{eq:power-constr}, 
the optimal solution is given by
\be
\E_\text a^\text{opt}(\nu) =  \frac{\sqrt{\pot_\text{a}}}{\| K_\text{r}\|} \, K_\text{r}^*(\nu) 
\label{eq:opt-Ea}
\ee
where $K_\text{r}(\nu)$ can be analytically obtained from Proposition~\ref{prop:1}.
This choice yields the maximum intensity at the receiver
\be
\intens_\text{r}^\text{max}= \pot_\text{a} \, \| K_\text{r}\|^2 \:.
\ee
Up to a scalar constant, the optimal solution \eqref{eq:opt-Ea} is the {\em phase-conjugated kernel}, i.e.,  the {\em matched filter} that retrofocuses energy onto $R$. 
Physically, this corresponds to the field distribution at the aperture that would be 
observed if a point source were placed at $R$ and back-propagated to the aperture plane. 
Strictly speaking, feeding the aperture with $\E_\text a^\text{opt}(\nu) $  ensures
constructive interference at the receiver point.

At this point, it is interesting to investigate the EM feature of
the wave radiated from the aperture when the aperture 
field distribution is given by \eqref{eq:opt-Ea}. The
starting point of this study consists of replacing 
$\E_\text a(\nu)$ in the  Huygens-Fresnel diffraction formula
\eqref{eq:int} with $\E_\text a^\text{opt}(\nu)$.
Let the optimal aperture field distribution be
decomposed as $\E_\text{a}^\text{opt}(\nu)= A_\text{a}^\text{opt}(\nu) \, 
e^{-j \Phi_\text a^\text{opt}(\nu)}$.
An approximate evaluation of the Huygens-Fresnel diffraction 
integral can be obtained via the stationary-phase method \cite{Wong-2001}, which states 
that the dominant contribution arises from points where the derivative of 
the total phase of the integrand 
\be
Q_\nu(z,x) \eqdef \Phi_\text a^\text{opt}(\nu)+
\frac{k_0}{2 z} (x-\nu)^2.
\label{eq:Q}
\ee
with respect to $\nu$, vanishes. Mathematically, the condition for phase stationarity 
$\tfrac{\partial \, Q_\nu(z,x)}{\partial \nu} =0$ yields
\be
x_\nu(z) = \nu + \frac{z}{k_0} \, \frac{\rm d}{\rm d  \nu} \Phi_\text a^\text{opt}(\nu) \:, 
\quad \text{for $\nu\in \mathcal{A}$} \: .
\label{eq:rays}
\ee
This equation describes a family of rays, each parameterized 
by the transverse coordinate $\nu$ 
on the aperture. 
The envelope of the family of curves described by \eqref{eq:rays}
defines a {\em caustic}, i.e., the locus where neighboring rays 
intersect or coalesce. This occurs at points  $(z,x)$ where
\be
\frac{\partial \, x_\nu(z)}{\partial \nu}  = 0
\quad \Leftrightarrow \quad 1+ \frac{z}{k_0} \, \frac{\rm d^2 \, \Phi_\text a^\text{opt}(\nu)}{\rm d^2  \nu} =0
\label{eq:second-der}
\ee
Solving equation \eqref{eq:second-der} gives $\nu_\text{c}(z)$, which represents
the critical aperture point that generates the caustic at propagation distance $z$.
The caustic curve is then obtained from \eqref{eq:rays} as 
$x_\text{c}(z) \eqdef x_{\nu_\text{c}(z)}(z)$.
In wave optics, the function $x_\text{c}(z)$ corresponds to the {\em bright trajectory} 
or ``backbone" of the beam,  over which multiple rays interfere constructively.

The derivation of the caustic $x_\text{c}(z)$ requires the explicit expression of 
$\Phi_\text a^\text{opt}(\nu)$. It can be inferred from \eqref{eq:opt-Ea} 
that, for a fixed receiver point,  the phase $-\Phi_\text a^\text{opt}(\nu)$ of  the optimal aperture field is the negative of the phase of the kernel $K_\text{r}(\nu)$.
On the basis of Proposition~\ref{prop:1}, the  knife-edge factor $F(z_\text{r},x_\text{r},\nu)$
(i.e., the error-function term due to the knife edge obstacle) is a slowly varying function, which 
acts primarily as a smooth amplitude apodization, thus yielding a negligible contribution to 
the phase of $K_\text{r}(\nu)$. Basically, any phase variation from $F(z_\text{r},x_\text{r},\nu)$ 
is confined to a narrow Fresnel transition near the geometric shadow boundary. Therefore, the rapidly varying factor of 
$K_\text{r}(\nu)$ turns out to be represented by the complex exponential
$e^{-j \, \frac{k_0}{2 \, z_\text{r}} \, (\nu-x_\text{r})^2}$. Consequently, one may write
\be
\Phi_\text a^\text{opt}(\nu) \approx  - \frac{k_0}{2 \, z_\text{r}} \, (\nu-x_\text{r})^2 + \text{const.}
\ee
which is  quadratic in $\nu$ and, hence, its second-order derivative 
\be
\frac{\rm d^2 \, \Phi_\text a^\text{opt}(\nu)}{\rm d^2  \nu} =- \frac{k_0 }{ z_\text{r}}
\label{eq:2-der}
\ee
is constant. By substituting \eqref{eq:2-der} in the envelope condition
\eqref{eq:second-der}, one gets $z=z_\text{r}$. The corresponding caustic reads 
as $x_\text{c}(z)=x_\text{r}$. We may finally infer that all rays converge to the single 
point $R$ and the caustic degenerates into a fold point, i.e., a focus. Therefore, the optimal field
distribution \eqref{eq:opt-Ea} generates a {\em focusing beam}, i.e., 
energy is focused constructively at the point $R$ where the receiver is located.

\begin{figure}[t]
\centering
\begin{subfigure}[b]{0.235\textwidth}
\centering
\includegraphics[width=\textwidth]{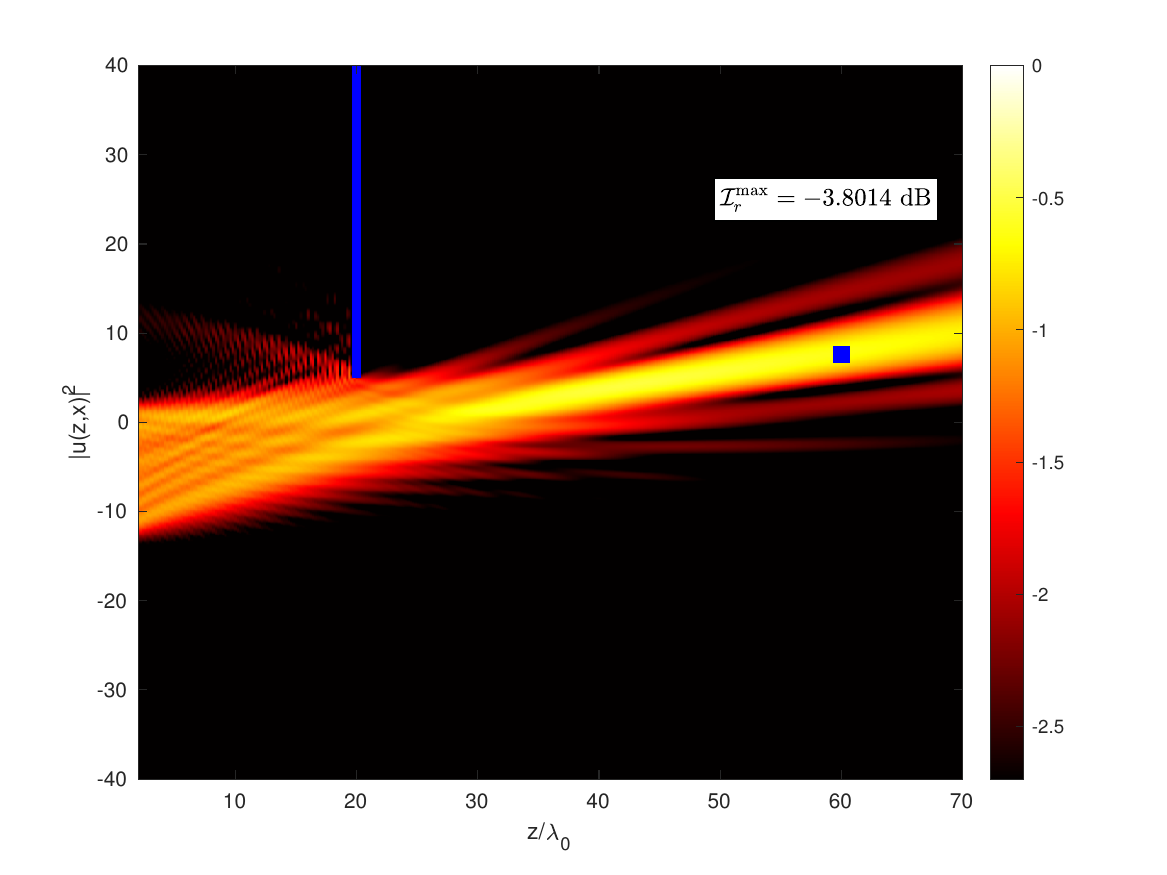}
\caption[]%
{{\small $x_\text{b}^{(1)}=5 \, \lambda_0$}}    
\label{fig:fig_2_11}
\end{subfigure}
\begin{subfigure}[b]{0.235\textwidth}  
\centering 
\includegraphics[width=\textwidth]{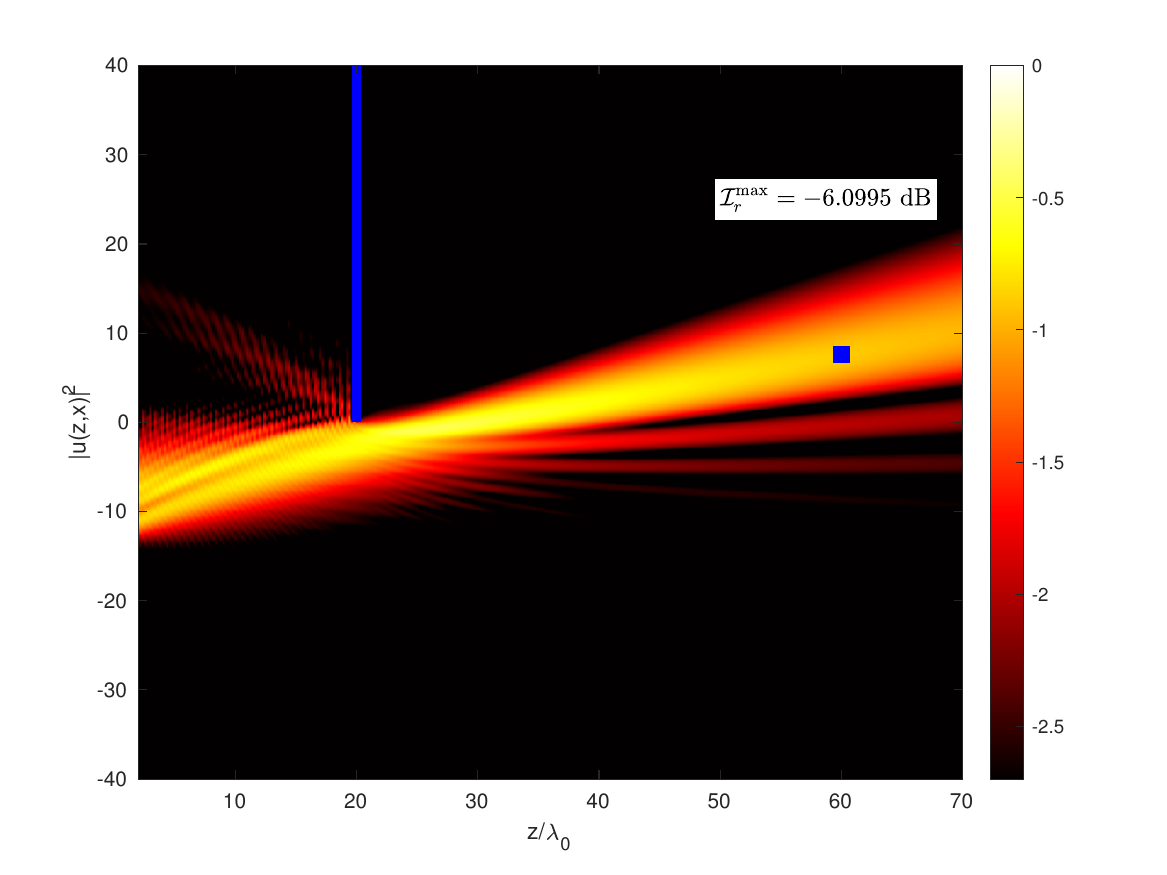}
\caption[]%
{{\small $x_\text{b}^{(1)}=0$}}    
\label{fig:fig_2_12}
\end{subfigure}
\vskip\baselineskip
\begin{subfigure}[b]{0.235\textwidth}   
\centering 
\includegraphics[width=\textwidth]{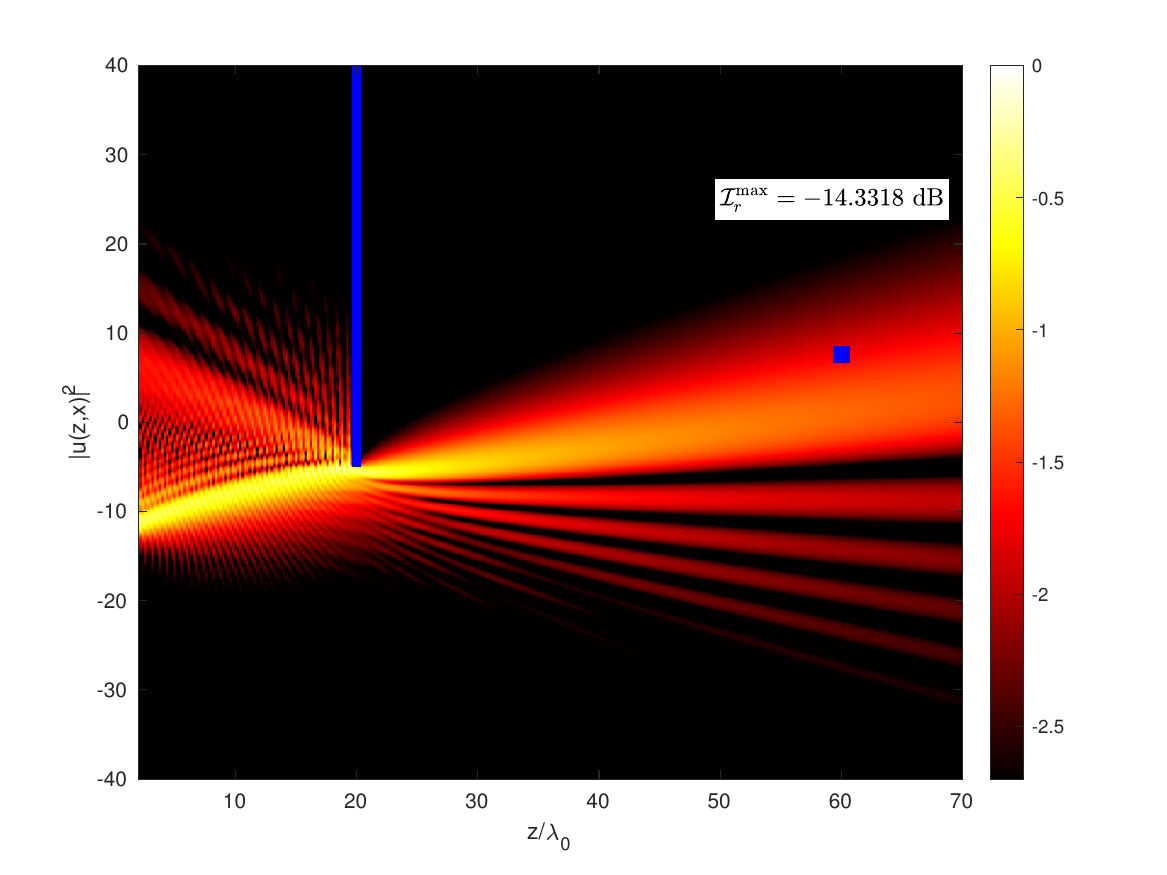}
\caption[]%
{{\small $x_\text{b}^{(1)}=-5 \, \lambda_0$}}    
\label{fig:fig_2_21}
\end{subfigure}
\begin{subfigure}[b]{0.235\textwidth}   
\centering 
\includegraphics[width=\textwidth]{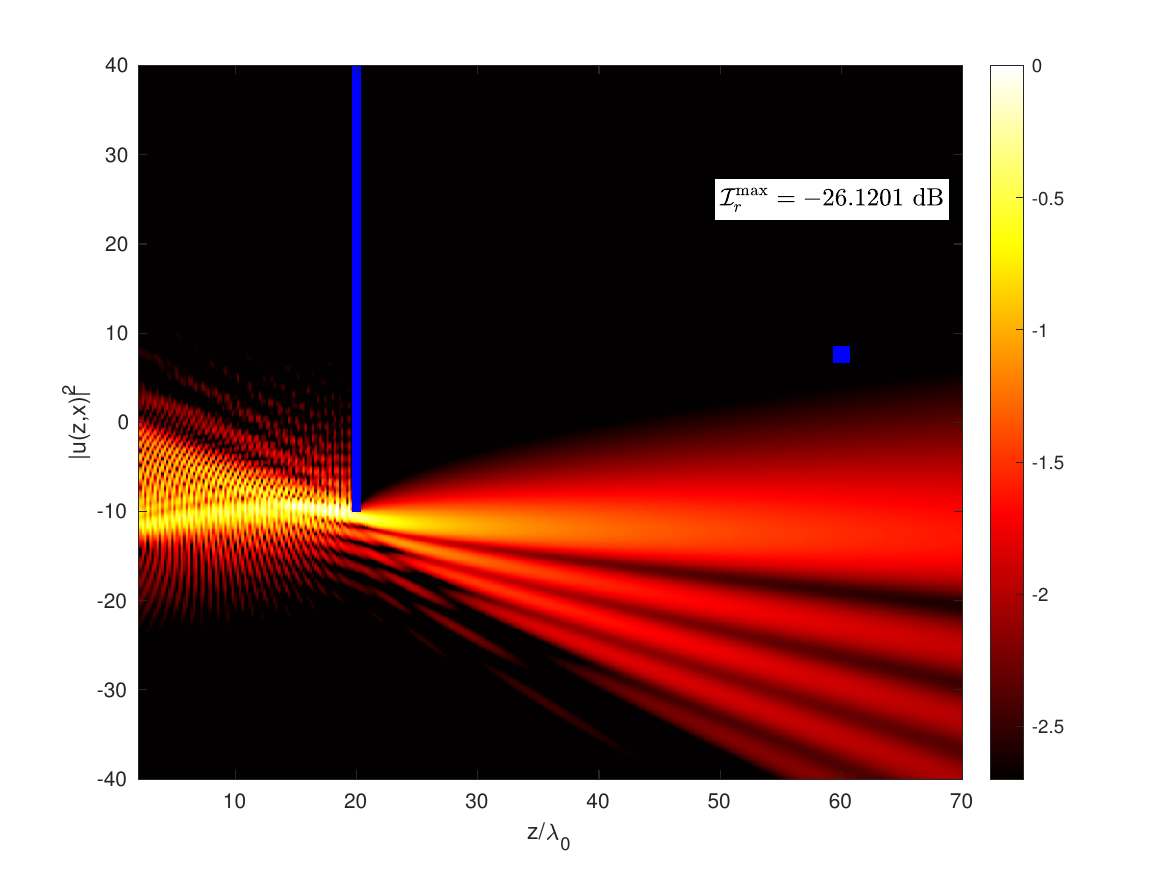}
\caption[]%
{{\small $x_\text{b}^{(1)}=-10 \, \lambda_0$}}    
\label{fig:fig_2_22}
\end{subfigure}
\caption[]%
{\small Intensity distribution (in $log_{10}$ scale) of the transmitted EM field as a function of $z$ and $x$.
The aperture extends from  $x_\text{a}^{(1)}=-13 \, \lambda_0$ to 
$x_\text{a}^{(2)}=2 \, \lambda_0$, partially obstructed by a semi-infinite
knife-edge obstacle  located at distance $z_\text{b}=20 \, \lambda_0$ from the aperture, 
for different values of $x_\text{b}^{(1)}$.
The blue thick line represents the obstacle, and the blue square indicates the receiver located 
at  $(60 \, \lambda_0, 7.5 \, \lambda_0)$.
Spatial coordinates are normalized with respect to $\lambda_0$.} 
\label{fig:fig_2}
\vspace{-5mm}
\end{figure}

\vspace{-2mm}
\section{Numerical results}
\vspace{-2mm}

Fig.~\ref{fig:fig_2}, shows the
intensity of the EM field radiated by the transmit aperture
as a function of $(z, x)$.
In this numerical example, the obstruction is modeled as a semi-infinite knife-edge obstacle, i.e., 
$x_\text{b}^{(2)}=+\infty$, which is located at a distance $z_\text{b}=20 \, \lambda_0$ from the aperture. The finite-size aperture is defined by 
$x_\text{a}^{(1)}=-13 \, \lambda_0$ and $x_\text{a}^{(2)}=2 \, \lambda_0$, which yields an overall  aperture width of   $\Delta x_\text{a} = 15 \, \lambda_0$.
The receiver is positioned at coordinates $(z_\text{r},x_\text{r})
=(60 \, \lambda_0, 7.5 \, \lambda_0)$. 
The position of the bottom edge of the obstacle takes on the values 
$x_\text{b}^{(1)}=5 \, \lambda_0$ in Fig.~\ref{fig:fig_2_11}, 
$x_\text{b}^{(1)}=0$ in Fig.~\ref{fig:fig_2_12}, 
$x_\text{b}^{(1)}=-5 \, \lambda_0$ in Fig.~\ref{fig:fig_2_21},
and $x_\text{b}^{(1)}=-10 \, \lambda_0$ in Fig.~\ref{fig:fig_2_22}.
In Fig.~\ref{fig:fig_2_11}, the obstacle does not 	obstructs the transmit aperture, 
whereas  the aperture is partially dimmed of  about $13 \%$, $47 \%$, and $80 \%$
in Fig.~\ref{fig:fig_2_12}, Fig.~\ref{fig:fig_2_21}, and Fig.~\ref{fig:fig_2_22},
respectively.
It can be observed from Fig.~\ref{fig:fig_2_11} that,  
when the aperture is unobstructed, the optimal aperture field \eqref{eq:opt-Ea}
yields coherent phasing of all aperture contributions at $R$.
On the other hand,  when part of the aperture is obstructed, 
the optimal aperture field is still a matched filter, which is 
now matched to the diffracted field generated by the 
unobstructed part of the aperture.
Henceforth, in Fig.~\ref{fig:fig_2_12}, Fig.~\ref{fig:fig_2_21}, and Fig.~\ref{fig:fig_2_22},
the field intensity at the receiver is reduced compared to the unobstructed 
case due to the limited visible aperture.  Nevertheless, thanks to the matched-filter structure
of $\E_\text a^\text{opt}(\nu)$, beam focusing remains effective even under 
partial obstruction of the aperture as long as a non-negligible fraction of the aperture 
remains visible and contributes coherently at the receiver.

\vspace{-2mm}
\section{Conclusions}
\label{sec:conc}
\vspace{-2mm}

This work presented a framework for optimal aperture field synthesis in near-field communications under partial LoS obstruction. Using a physically consistent knife-edge diffraction model, the propagation channel was formulated as a linear operator, with the optimal solution given by a matched filter to the diffraction kernel. Numerical results showed that beam focusing remains effective despite partial obstruction, provided that the visible aperture is sufficiently large. 
Our formulation provides a link between wave propagation, signal processing, and hardware design, offering a way for efficient beamforming with continuous apertures in NLoS channels.


\end{document}